\documentclass[runningheads]{llncs}

\usepackage[T1]{fontenc}
\usepackage{graphicx}
\usepackage{amsmath,amssymb}
\usepackage{bussproofs}
\usepackage{ulem}
\usepackage{stmaryrd}
\usepackage{booktabs}
\usepackage{makecell}
\usepackage{appendix}
\usepackage{mathtools}

\newcommand{\Set}{\mathrm{Set}}
\newcommand{\Term}{\mathrm{Term}}
\newcommand{\Instance}{\mathrm{Instance}}
\newcommand{\Event}{\mathrm{Event}}
\newcommand{\Nat}{\mathbb{N}}

\newcommand{\tuple}[1]{\langle #1 \rangle}
\newcommand{\state}[1]{\tuple{#1}}

\usepackage{relsize} 
\usepackage{xspace}

\newcommand{\cc}[1]{\mbox{\smaller[0.5]\texttt{#1}}\xspace}

\newcommand{\unsat}{\cc{unsat}}
\newcommand{\timeout}{\cc{timeout}}
\newcommand{\unknown}{\cc{unknown}}

\newcommand{\idealSolverID}{\mbox{\sf Id}}
\newcommand{\idealSolver}{\idealSolverID \xspace}
\newcommand{\idealSolverSig}{\Sigma_{\idealSolverID}}

\newcommand{\naiveSolverID}{\mbox{\sf N}}
\newcommand{\naiveSolver}{\naiveSolverID \xspace}
\newcommand{\naiveSolverSig}{\Sigma_{\naiveSolverID}}

\newcommand{\trackingSolverID}{\mbox{\sf T}}
\newcommand{\trackingSolver}{\trackingSolverID \xspace}
\newcommand{\trackingSolverSig}{\Sigma_{\trackingSolverID}}

\newcommand{\cgcSolverID}{\mbox{\sf C}}
\newcommand{\cgcSolverSig}{\Sigma_{\cgcSolverID}}
\newcommand{\cgcSolver}{\cgcSolverID \xspace}

\newcommand{\derivable}[4]{#3 \vdash_{#1,#2} #4}
\newcommand{\derivableEI}[2]{\derivable{E}{I}{#1}{#2}}
\newcommand{\isSubterm}[2]{#1 \in #2}
\newcommand{\weight}[1]{\mathit{w}(#1)}

\newcommand{\maxCost}{\mbox{\sf maxG}}

\newcommand{\instantiatedTerm}{\theta}
\newcommand{\instUniv}{I^\infty}

\newcommand{\genParam}[3]{\mathrm{gen}^\mathsf{#1}_{#2}(#3)}
\newcommand{\genideal}[2]{\genParam{ideal}{#1}{#2}}

\newcommand{\Inst}{$\mathit{Inst}$ \xspace}
\newcommand{\Congruence}{$\mathit{Congruence}$ \xspace}
\newcommand{\Input}{$\mathit{Input}$ \xspace}

\newcommand{\sIdeal}{\state{\tau}}
\newcommand{\eqEvent}[1]{\mathtt{eq}(#1)}
\newcommand{\eqMergeEvent}[3]{\mathtt{eq}(#1,#2,#3)}
\newcommand{\instEvent}[1]{\mathtt{inst}(#1)}
\newcommand{\emptySeq}{\langle \rangle}
\newcommand{\enabledParam}[3]{\varepsilon^\mathsf{#1}_{#2}(#3)}
\newcommand{\enabled}[2]{\enabledParam{}{#1}{#2}}
\newcommand{\transition}[3]{#1 \to_{\text{#2}} #3}
\newcommand{\reachable}[3]{#1 \leadsto_{\text{#2}} #3}
\newcommand{\transitionSup}[4]{#1 \to_{\text{#2}}^{\text{#3}} #4}
\newcommand{\reachableSup}[4]{#1 \leadsto_{\text{#2}}^{\text{#3}} #4}
\newcommand{\tIdeal}[2]{\transition{#1}{$\idealSolverID$}{#2}}
\newcommand{\tS}[3]{\transition{#2}{$#1$}{#3}}
\newcommand{\rS}[3]{\reachable{#2}{$#1$}{#3}}
\newcommand{\PushEq}{\mbox{PushEq}\xspace}
\newcommand{\PopEq}{\mbox{PopEq}\xspace}
\newcommand{\PushInst}{\mbox{PushInst}\xspace}
\newcommand{\PopInst}{\mbox{PopInst}\xspace}

\newcommand{\sNaive}{\sIdeal}
\newcommand{\tNaive}[2]{\transition{#1}{\naiveSolverID}{#2}}

\newcommand{\genMap}{g}
\newcommand{\genTracked}[3]{\genParam{\trackingSolverID}{#1,#2}{#3}}
\newcommand{\lemmaRef}[1]{Lemma~\ref{lemma:#1}}
\newcommand{\enabledTr}[2]{\enabledParam{\trackingSolverID}{#1}{#2}}
\newcommand{\updateParam}[4]{\mathrm{update}^\mathsf{#1}_{#2}(#3, #4)}
\newcommand{\updateTr}[3]{\updateParam{\trackingSolverID}{#1}{#2}{#3}}
\newcommand{\sTr}{\state{\tau, \genMap}}
\newcommand{\tTr}[2]{\transition{#1}{\trackingSolverID}{#2}}

\newcommand{\genMapCGC}{\gamma}
\newcommand{\cgcMap}{\kappa}
\newcommand{\stickyUpdates}{\delta}
\newcommand{\sCGC}{\state{\tau, \cgcMap, \genMapCGC, \stickyUpdates}}
\newcommand{\tCGCMerge}[3]{\transition{#1}{$CGC(#2)$}{#3}}

\newcommand{\tApplyUpdates}[4]{\transition{#3}{$\text{replay}_{#1}(#2)$}{#4}}
\newcommand{\enabledCGC}[2]{\enabledParam{\cgcSolverID}{#1}{#2}}
\newcommand{\genCGC}[3]{\genParam{\cgcSolverID}{#1,#2}{#3}}
\newcommand{\updateCGC}[3]{\updateParam{\cgcSolverID}{#1}{#2}{#3}}
\newcommand{\newStickyUpdates}[4]{\Delta_{#1, #2}(#3, #4)}
\newcommand{\tCGC}[2]{\transition{#1}{$\cgcSolverID$}{#2}}

\newcommand{\mapMultiUpdate}[4]{#1[#2 \mapsto #3 \mid #4]}
\newcommand{\defineEq}{:=}
\newcommand{\cgParam}[3]{\mathit{congruent}_{#1}(#2, #3)}
\newcommand{\cg}[2]{\cgParam{E}{#1}{#2}}

\newcommand{\qInst}[2]{#1\llbracket#2\rrbracket}

\newif\ifdraft
\ifdraft
\usepackage{color}
\newcommand{\can}[1]{\textcolor{red}{\fbox{CC:} #1}}
\else
\newcommand{\can}[1]{}
\fi

\newcommand{\secref}[1]{\S\ref{s:#1}}

\usepackage[most]{tcolorbox} 
\usepackage{xcolor}          

\definecolor{exBoxFill}{RGB}{255, 242, 193} 

\newtcolorbox{exampleBox}{
    colback=exBoxFill,    
    boxrule=0.8pt,      
    arc=2mm,            
    left=5pt, right=5pt, top=5pt, bottom=5pt, 
    width=\textwidth,   
    enlarge left by=0mm,
    enlarge right by=0mm
}

\begin{document}
\title{A Deeper Look at Depth: Stable Generation Accounting for Quantifier Reasoning}
\titlerunning{A Deeper Look at Depth}

\author{
Can Cebeci\inst{1}
\and Nikolaj Bjørner\inst{2}
\and George Candea\inst{1}
\and Clément Pit-Claudel\inst{1}
}

\institute{
EPFL, Lausanne, Switzerland
\and
Microsoft Research, Redmond, USA
}

\maketitle              
\begin{abstract}
SMT solvers make automated verification convenient. At the same time, solvers suffer from instability, whereby seemingly inconsequential changes to the input may cause a previously quickly produced proof to fail or time out. This paper addresses a common cause of outcome instability (i.e., unsat/unknown fluctuations) in the context of program verification.
We demonstrate that the generation (i.e., depth) accounting used to make quantifier instantiation practical and implemented in multiple state-of-the-art SMT solvers is non-confluent (i.e., prone to divergence), and that this leads to instability. We then address this deficiency and develop a new accounting method that is stable. The new method is justified using a sequence of refinements from an abstract, confluent solver model all the way to our implementation in Z3.
Our empirical evaluation demonstrates that our implementation reduces outcome instability by 94\% in the unstable core of the Mariposa benchmark without leading to performance regressions.
\keywords{SMT solving \and
instability \and
quantifier reasoning}
\end{abstract}

\section{Introduction}

SMT solvers make automated verification convenient but brittle.
They automatically discharge most proof obligations, significantly lowering the required effort and expertise for developing proofs.
At the same time, solvers suffer from instability, whereby seemingly inconsequential changes to the input may cause a previously quickly produced proof to fail or time out~\cite{mariposa}.
Recent work~\cite{instability-conjecture} conjectures that the instability users encounter in practice is addressable, not fundamental. This paper addresses a common cause of outcome instability (i.e., \unsat queries turning \unknown) in the context of program verification.

Controlling quantifier instantiation is crucial for stability~\cite{shake:fmcad2024,cazamariposas,trigger-selection:cav2016}.
This is because quantifiers are heavily used to encode verification conditions, and the solver's success depends on choosing the right instances:
under-instantiation leads to \unknown and over-instantiation leads to \timeout.

Program verifiers typically control instantiation by configuring the solver to instantiate based on E-matching~\cite{efficient-ematching} exclusively, and up to some form of depth threshold~\cite{dafny,verus:sosp24,fstar}. For Z3~\cite{Z3}, they set \cc{smt.qi\_eager\_threshold}. For CVC5~\cite{cvc5}, \cc{inst-max-level}. Following Z3's terminology, we refer to measures of a term's instantiation depth as its \emph{generation number}.

We demonstrate that the generation accounting implemented in some of the state-of-the-art SMT solvers is non-confluent (i.e., prone to divergence), and that this leads to instability.
More precisely, we show that in certain cases a term's generation number (and subsequently the solver's success) is dependent on the order in which the solver explores proof branches.
These cases are not simple solver bugs.
Rather, the ambiguity is conceptual: to the best of our knowledge, there is no precise definition of what generation numbers should represent.
We use the terminology \emph{confluence} here in the sense of Church-Rosser~\cite{ChurchRosser} properties. 
It is fundamental in automated deduction, where a terminating and confluent set of rewrite rules
ensures decidability~\cite{KnuthBendix,BuchbergerGrobner} in several settings. SMT solvers using E-matching for quantifier instantiation ensure termination by bounding quantifier generation.
But as we observe, bounding instantiation using generation thresholds is an important ingredient for breaking confluence. 

We address this issue with the following contributions:
\begin{itemize}
    \item A precise, solver-agnostic specification of generation semantics for E-matching with quantifier weights.
    \item \idealSolver: an abstract state machine model of solver behavior with access to idealized generation numbers. \idealSolver is confluent, but not practically implementable. It serves as a reference for the generation accounting problem independently of a specific implementation.
    \item \naiveSolver, \trackingSolver, \cgcSolver: a series of implementable, increasingly efficient solver models, each obtained by refining the previous one. We prove, through bisimulation, that each solver is confluent and no more permissive than \idealSolver with quantifier instantiation.
    \item An efficient implementation of \cgcSolver in Z3, upstream as of version 5.0.
    \item An empirical evaluation demonstrating that the implementation reduces outcome instability by 94\%, while incurring 2\% fewer timeouts, in the unstable core of the Mariposa benchmark~\cite{mariposa}, which includes program verification queries. We also demonstrate that our implementation does not lead to performance regressions for SMT-LIB benchmarks~\cite{smtlib2025}.
\end{itemize}

In this work-in-progress paper, we present proof sketches for the refinements \naiveSolver, \trackingSolver, and \cgcSolver, aiming to convey the underlying intuition rather than full rigor. We present some of the fully formalized proofs in the appendix; the correctness of the remaining refinements is conjectured, supported by the proof sketches.

\section{Background and Motivating Examples}
In this section, we provide an overview of E-matching with generation thresholds and quantifier weights, and motivate the need to precisely define generation semantics. 

\vspace{35pt}

We use the following query as a running example:
\begin{exampleBox}
\begin{verbatim}
forall x. f(x+1) > f(x)
f(0) > 0
f(3) < 0
\end{verbatim}
\end{exampleBox}
The query is unsatisfiable: instantiating the quantifier three times with $x\defineEq0,1,2$ leads to a contradiction.

\subsection{E-matching}
Program verifiers typically configure SMT solvers to produce quantifier instantiations using E-matching, whereby the solver produces instances based on ground terms that match a given pattern.
For our running example, the pattern $f(x)$ would lead to the right instantiations: $f(0)$ in the input would match $f(x)$ with $x\defineEq0$. That instantiation would create the term $f(1)$, which would match again with $x:=1$ and so on.

Though the pattern is productive in this scenario, it also creates a matching loop. Had the query been satisfiable (e.g., had the last assertion been omitted), the solver could produce infinitely many instances rather than quickly giving up. To guard against this, most solvers expose configuration parameters that set thresholds controlling when to give up and return \unknown.

\subsection{Generation thresholds}
A common idea is to threshold the \textit{depth} of instantiation chains as opposed to the total number of instantiations. The latter is prone to brittleness, as a productive instantiation may never take place depending on the order in which the solver explores proof branches.

One way to track instantiation depth is to associate a generation number with each term, such that all input terms have generation $0$, and terms created through an instantiation have $1$ plus the maximal generation among matching terms~\cite{z3-internals}. We block an instantiation if the maximal generation is above the threshold.  In the example above $f(0)$ would have generation $0$, $f(1)$ would have $1$ and so on. Therefore, a generation threshold of $3$ would suffice for the solver to return \unsat.

\subsection{Quantifier weights}
Now consider adding another quantifier to the query:
\begin{exampleBox}
\begin{verbatim}
forall x. f(f(x)) > f(x+2) > f(x) ::: pattern f(x)
\end{verbatim}
\end{exampleBox}
This quantifier is more expensive: it creates two fresh terms, and both terms re-match the pattern.
So it produces $2^{c-1}$ instantiations where $c$ is the generation threshold.
We would like to keep instantiation chains involving this quantifier shorter than others.
One way to do this is to associate a larger \textit{weight} with expensive quantifiers, and increment the generation numbers of fresh terms by the weight, instead of just by $1$.

To prevent recursive matches on the expensive quantifier while still allowing the productive instantiations, we can configure the query as follows:

\begin{exampleBox}
\begin{verbatim}
generation threshold: 3
forall x. f(f(x)) > f(x+2) > f(x) ::: pattern: f(x), weight: 3
forall x. f(x+1) > f(x) ::: pattern: f(x), weight: 1
f(0) > 0
f(3) < 0
\end{verbatim}
\end{exampleBox}

\subsection{Generation ambiguity}

\subsubsection{Example 1}
In the current state of our working example, the generation number of $f(2)$ is ambiguous:
$f(0)$ is part of the input so it has generation $0$.
Then the first quantifier can generate $f(2)$ at generation $3$. On the other hand the second quantifier can generate $f(2)$ at generation $2$ after two rounds of instantiation.
This ambiguity gives rise to instability, since $f(2)$'s generation determines whether the instantiation producing $f(3) > f(2)$ is allowed, and consequently whether the query is \unsat or \unknown.

\subsubsection{Example 2}
The ambiguity is not uniquely brought on by quantifier weights. Consider the following query:
\begin{exampleBox}
\begin{verbatim}
generation threshold: 3
forall x. g(x) == h(x) ::: pattern g(x), weight 1
forall x. h(x) == g(x+1) ::: pattern h(x), weight 1
forall x. f(g(x)) > g(x) ::: pattern g(x), weight 1
forall x. f(f(x)) > f(x) ::: pattern f(x), weight 1
forall x. f(f(x)) < 0 ::: pattern f(f(x)), weight 1
g(0) > 0
\end{verbatim}
\end{exampleBox}
$g(0)$ has generation $0$.
By instantiating the first two quantifiers, we can generate $g(1)$ at generation $2$, and learn $g(0) = g(1)$.
From the third quantifier, we get $f(g(0))$ at generation $1$.

The fourth quantifier should either match $f(g(0))$ or $f(g(1))$, the latter through the equality we learned. Matching both would be wasteful, since $x := g(0)$ and $x := g(1)$ are equivalent bindings. If the fourth quantifier matches $f(g(0))$, we get $f(f(g(0)))$ at generation $2$, which matches the fifth quantifier and leads to a contradiction, returning \unsat. But if the fourth quantifier matches $f(g(1))$, since $g(1)$ has generation $2$ , we get $f(f(g(1)))$, an equivalent term, at generation $3$ and so give up, returning \unknown, without instantiating the fifth quantifier.

\section{Preliminaries}
\subsection{Terms and contexts}
Let $x$, $c$, and $f$ range over variables, constants, and uninterpreted function
symbols, respectively. Terms are given by the grammar
\[
t \defineEq
    x
    \mid c
    \mid f(t_1,\ldots,t_n).
\]

A context is a term containing exactly one hole~$\square$, which is a
proper subterm. Contexts are defined recursively by
\[
C \defineEq
    f(t_1,\ldots,\square,\ldots,t_n)
    \mid
    f(t_1,\ldots,C,\ldots,t_n).
\]

For a context $C$ and term $t$, we write $C[t]$ for the term obtained by
replacing the unique occurrence of $\square$ in $C$ with $t$.

\subsection{Congruence and equality}
We make a central distinction between equivalence and congruence modulo the theory of equality. Equivalence is a weaker notion:
two expressions $s, t$, are equivalent in $E$, written
$t \cong_E s$, if the equality between $s$ and $t$ can be derived from equalities in $E$ using any rule in the theory of equality (reflexivity, symmetry, transitivity, congruence).
On the other hand, two terms are congruent if they are equivalent \emph{and} their equivalence can be inferred using the congruence
rule:

\begin{definition}[Congruence between terms]
We say that $t$ and $t'$ are congruent if they have the same function symbol and all arguments are equivalent in $E$.
In other words, for every function symbol $f$, congruence is given by
\[
\AxiomC{$t_i \cong_E t_i'$ for $i = 1, \ldots, n$}
\UnaryInfC{$\cg{f(t_1, \ldots, t_n)}{f(t_1', \ldots, t_n')}$}
\DisplayProof
\]
\end{definition}
The E-matching and congruence-closure literature~\cite{fast-congruence-closure} uses terminology that can blur this distinction. For example, the statement ``if two terms $f(t_1,\ldots,t_n)$ and $f(t'_1,\ldots,t'_n)$ are congruent, then it is wasteful to try to match both''~\cite{efficient-ematching} is not valid under the weaker interpretation where only $f(a)=f(b)$ is known: matching $f(x)$ can produce non-equivalent bindings $x:=a$, $x:=b$ when $a \neq b$.

\subsection{Quantifiers and instances}
Quantifiers and quantifier instances are not terms. Instead, we treat them as
distinct syntactic objects.

For simplicity, we assume each quantifier binds a single pattern and that quantifiers are not nested. Generalizing to multiple bound variables and nested quantifiers is conceptually straightforward but at the cost of more involved formalization.

A quantifier instance is denoted by $\qInst{q}{\instantiatedTerm}$, where $q$ is a quantifier
and $\instantiatedTerm$ is the term matching its pattern (not the bound variable).
We write
$\isSubterm{t}{\qInst{q}{\instantiatedTerm}}$ to mean that the instantiated body of $q$ with $\instantiatedTerm$ contains $t$ as a subterm.

Each quantifier $q$ is associated with a non-negative integer weight, denoted by $\weight{q}$.

\subsection{Other notation}
\begin{itemize}
    \item For functions $f,g,h$ and predicate $P$, $f' = f[g(x) \mapsto h(x) \mid P(g(x))]$ denotes $\forall x. \; (P(g(x)) \implies f'(g(x)) = h(x))$
    $\land \forall y .\; (\lnot P(y) \implies f'(y) = f(y))$
    \item Similarly $f' = f[y \mapsto z]$ denotes  $f'(y) = z \land \forall x. \; x \neq y \implies f'(x) = f(x)$
\end{itemize}

\section{Defining Generation Semantics}
Intuitively, we define the generation number of a term as the weighted depth of the shallowest derivation tree, where leaves are input terms and edges are quantifier instantiations or equality rewrites using previously-derived terms.

For intuition, consider the following scenario. Solid lines represent instantiations and dashed lines represent rewrites:
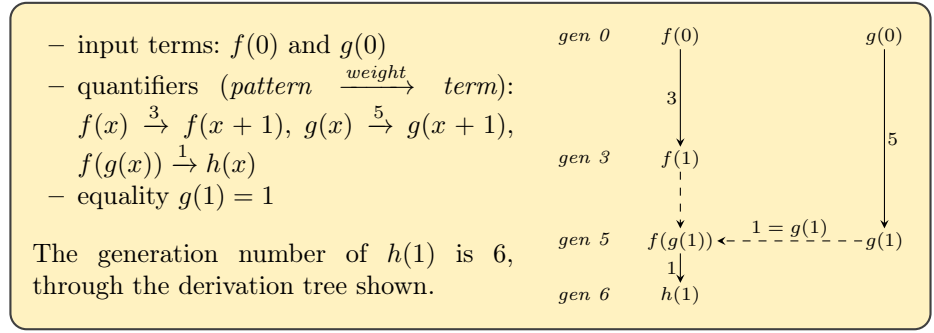
\begin{exampleBox}
\begin{minipage}[c]{0.55\textwidth}
\begin{itemize}
    \item input terms: $f(0)$ and $g(0)$
    \item quantifiers
    ($\textit{pattern}\xrightarrow{\textit{weight}}\textit{term}$): $f(x)\xrightarrow{3}f(x+1)$, $g(x)\xrightarrow{5}g(x+1)$, $f(g(x))\xrightarrow{1}h(x)$
    \item equality $g(1) = 1$
\end{itemize}
The generation number of $h(1)$ is 6, through the derivation tree shown.
\end{minipage}
\hfill
\begin{minipage}[c]{0.48\textwidth}
\centering
\begin{tikzpicture}[
    every node/.style={font=\scriptsize, inner sep=1pt},
    >=stealth, scale=0.9
]
    \node (f0) at (0,0)   {$f(0)$};
    \node (g0) at (3.0,0) {$g(0)$};

    \node (f1) at (0,-1.8) {$f(1)$};

    \node (fg1) at (0,-3.0)   {$f(g(1))$};
    \node (g1)  at (3.0,-3.0) {$g(1)$};

    \node (h1) at (0,-3.8) {$h(1)$};

    \foreach \g/\y in {0/0, 3/-1.8, 5/-3.0, 6/-3.8}
        \node[anchor=east, font=\scriptsize\itshape] at (-1,\y) {gen \g};

    \draw[->] (f0) -- node[left] {$3$} (f1);
    \draw[->] (g0) -- node[right] {$5$} (g1);
    \draw[dashed,->] (f1) -- (fg1);
    \draw[dashed,->] (g1) -- node[above] {$1=g(1)$} (fg1);
    \draw[->] (fg1) -- node[left] {$1$} (h1);
\end{tikzpicture}
\end{minipage}
\end{exampleBox}

\begin{definition}[Reachable generation]
Given a set of quantifiers $Q$, input terms $A$, and equalities $E$, and instantiations $I$, a reachable generation $g$ of a term $t$ is a tuple $\derivableEI{t}{g}$, that can be obtained from one of the inferences:
\[
\begin{array}{c}
\AxiomC{$t \in A$}
\RightLabel{(\text{\mbox{Input}})}
\UnaryInfC{$\derivableEI{t}{0}$}
\DisplayProof
\\
\\
\AxiomC{$q \in Q$}
\AxiomC{$\isSubterm{t}{\qInst{q}{\instantiatedTerm}}\in I$}
\AxiomC{$\derivableEI{\instantiatedTerm}{g}$}
\RightLabel{(\text{\mbox{Inst}})}
\TrinaryInfC{$\derivableEI{t}{\weight{q} + g}$}
\DisplayProof
\\
\\
\AxiomC{$s_1 \cong_E s_2$}
\AxiomC{$\derivableEI{t[s_1]}{g}$}
\AxiomC{$\derivableEI{s_2}{g'}$}
\RightLabel{(\text{\mbox{Congruence}})}
\TrinaryInfC{$\derivableEI{t[s_2]}{\max(g, g')}$}
\DisplayProof
\end{array}
\]
\end{definition}

$A$ and $Q$ are constant given an input query, so we assume them implicit arguments. $A$ is subterm-closed (i.e., $\forall t. \; t \in A \land \isSubterm{t'}{t} \implies t' \in A$)

It is not practically feasible for solver implementations to track generation numbers in a way that exhaustively considers all possible quantifier instantiations. We therefore make a distinction between ideal and implementation generations.

\newcommand{\gen}[3]{\genParam{}{#1,#2}{#3}}

\begin{definition}[Implementation generation]
We define the implementation generation of a term $t$ with respect to $A, Q, E, I$ as
\[
    \gen{E}{I}{t} \defineEq \min \{ g \mid \derivableEI{t}{g} \}
\]
\end{definition}

\begin{definition}[Universal set of instantiations]
\[
\instUniv \defineEq \{\qInst{q}{\instantiatedTerm} \mid
q \in Q, \instantiatedTerm \mbox{ is a term} \}
\]
\end{definition}

\begin{definition}[Ideal generation]
We define the ideal generation of a term $t$ with respect to $A, Q, E$ as
\[
\genideal{E}{t} \defineEq \min \{ g \mid \derivable{E}{\instUniv}{t}{g} \}
\]
\end{definition}
Whenever convenient, we denote $\genideal{E}{t}$ with $\gen{E}{\instUniv}{t}$

For a set of equalities $E$, $\genideal{E}{t}$ is uniquely determined and independent of instantiation history, backtracking, or internal data-structure state.

\section{An Ideal Solver Model - \idealSolver}
This section presents \idealSolver, a solver model with access to $\genideal{E}{t}$, used as a reference for practical algorithms. 

Our method only depends on inference steps involving quantifier instantiations and equalities, and the model omits preconditions and inferences that are irrelevant to us.
In particular, we do not model conflicts, instead assuming the solver may backtrack at any point.

\subsection{State representation}

The ideal solver maintains a state of the form
$
\idealSolverSig : \sIdeal
$
where
\begin{itemize}
    \item $\tau: \Event^*$ - stack of locally asserted equalities and instances
    \item $\Event \defineEq \eqEvent{t_1 \simeq t_2} \mid \instEvent{\qInst{q}{\instantiatedTerm}}$
    for $q \in Q$
\end{itemize}

The initial state is, $\sigma^{init}_{\idealSolverID} \defineEq \state{\tau\defineEq\emptySeq}$.

Beyond the current state $\state{\tau}$ we assume the following global constants:
\begin{itemize}
    \item $A: \Set(\Term)$ (ground terms in input query),
    \item $Q: \Set(\mathrm{Quantifier})$ (quantifiers),
    \item $\maxCost: \Nat$ (generation threshold).
\end{itemize}

For convenience we also define a few shortcuts based on a state:
\begin{itemize}
    \item $E(\tau) \defineEq \{ e \mid \eqEvent{e} \in \tau \} : \Set(\Term)$ - set of equalities that are asserted within the scope.
    \item $I(\tau)\defineEq \{ \qInst{q}{\instantiatedTerm} \mid \instEvent{\qInst{q}{\instantiatedTerm}} \in \tau \} : \Set(\Instance)$ - set of active quantifier instantiations.
\end{itemize}

\newcommand{\stack}[1]{\mathcal{T}(#1)}

Unless indicated otherwise, we use $E$ as shorthand for $E(\tau)$ and $I$ for $I(\tau)$. We also use $\stack{\sigma}$ to refer to $\tau$ where $\sigma \defineEq \state{\tau}$.

\begin{definition}[Enabled instances]
For a set of equalities $E$ and quantifier instances $I$, define the set of instances enabled within generation threshold $\maxCost$ as
\[
\enabled{E}{I} \defineEq \{\qInst{q}{\instantiatedTerm} \mid \weight{q} + \gen{E}{I}{\instantiatedTerm} \leq \maxCost, \ q \in Q,\ \instantiatedTerm \mbox{ is a term}\}
\]
\end{definition}

\noindent
The transition rules for \idealSolver are as follows:
\[
\begin{tabular}{ccc}
\AxiomC{$e \defineEq (t_1 \simeq t_2)$}
\RightLabel{(\PushEq)}
\UnaryInfC{$\tIdeal{\sIdeal}{\state{\tau \cdot \eqEvent{e}}}$}
\DisplayProof
&\quad\quad &
\AxiomC{}
\RightLabel{(\PopEq)}
\UnaryInfC{$\tIdeal{\state{\tau \cdot \eqEvent{e}}}{\state{\tau}}$}
\DisplayProof
\\[2em]

\AxiomC{$\qInst{q}{\instantiatedTerm} \in \enabled{E}{\instUniv}$}
\RightLabel{(\PushInst)}
\UnaryInfC{$\tIdeal{\sIdeal}{\state{\tau \cdot \instEvent{\qInst{q}{\instantiatedTerm}}}}$}
\DisplayProof
& &
\AxiomC{}
\RightLabel{(\PopInst)}
\UnaryInfC{$\tIdeal{\state{\tau\cdot \instEvent{\qInst{q}{\instantiatedTerm}}}}{\state{\tau}}$}
\DisplayProof
\end{tabular}
\]

We use $\rS{\idealSolverID}{}{}$ to denote the transitive closure of $\tS{\idealSolverID}{}{}$ (i.e., reachability).

\subsection{Stability and confluence}
Our goal is to prove that each solver is stable, i.e., if there is some way for the solver to return \unsat, the solver will always do so. We do not model decisions and conflicts, or returning \unsat or \unknown. Instead we provide a definition of confluence and argue that a confluent solver is outcome-stable.

\newcommand{\observable}[1]{\mathrm{Obs}(#1)}

A solver returns \unsat if it finds a set of conflicts that suffice to prove false, and returns \unknown if no such set can be found. 
When the solver visits $\state{\tau}$, it may find a conflict between equalities in $E(\tau)$ considering instances in $I(\tau)$.
We define all such (potentially conflicting) contexts the solver may observe starting at state $\sigma$ as $\observable{\sigma}$:
\begin{definition} [Observable contexts]
\[
\observable{\sigma} \defineEq
    \{\tuple{E_o, I_o} \mid 
        \rS{S}{\sigma}{\sigma'} \land 
        E_o \subseteq E(\stack{\sigma'}) \land
        I_o \subseteq I(\stack{\sigma'})
        \}
\]

\end{definition}

We call a solver confluent if choosing to take certain steps never shrinks the set of observable contexts.
\begin{definition}[Confluence]
Solver $S$ with initial state $\sigma^{init}_S$ is confluent if for all states $\sigma$
\[
\rS{S}{\sigma^{init}_S}{\sigma} \implies \observable{\sigma^{init}_S} = \observable{\sigma}
\]
\end{definition}

A confluent solver is outcome-stable: if the solver may return \unsat, there is a set of conflicts, each observable from the initial state, that collectively prove \unsat. Due to confluence those conflicts remain observable after every transition, so the solver will never return \unknown.

\begin{theorem}
    \idealSolver is confluent.
\end{theorem}
\begin{proof}
    \idealSolver can perfectly backtrack every step, so it is always possible to return from $\sigma$ to $\sigma^{init}_S$, and from there observe the same contexts.
\end{proof}

\section{Realizable (And Efficient) Solver Models}
Here we present our algorithm through a sequence of realizable solver models that get iteratively more efficient. \secref{implementation} then describes how the last algorithm translates into a Z3 implementation.

For each solver $S$, we prove that $S$ is confluent and that $S$ is not more permissive with quantifier instantiation than \idealSolver. We do so by defining a simulation relation between solvers, proving bisimulation between \idealSolver and each subsequent solver, and proving lemmas which state that bisimulation implies the desired properties (\lemmaRef{sim-inherits-obs}, \lemmaRef{bisim-implies-confluence}).

The state tuple of each solver model includes $\tau$, the event trail. $\tau$ is always initially empty, and $\PushEq$ is always available. We define a notion of subsumption between trails, which we then use to define simulation between solvers:

\newcommand{\subsumes}[2]{#2 \sqsubseteq #1}

\begin{definition} [Trail subsumption]
We define subsumption between trails $\tau$, $\tau'$ as
\[
\subsumes{\tau'}{\tau} \defineEq E(\tau) \subseteq E(\tau') \land I(\tau) \subseteq I(\tau')
\]
\end{definition}

\newcommand{\simulates}[2]{#2 \preceq #1}

\begin{definition}[Simulation]
We define simulation between solvers $R$,$S$ as
\[
\begin{aligned}
\simulates{S}{R}
\defineEq {}&
\forall \sigma_S, \sigma_R, \sigma_R'.\;
    \rS{S}{\sigma_S^{init}}{\sigma_S}
    \land \rS{R}{\sigma_R^{init}}{\sigma_R}
    \land \subsumes{\stack{\sigma_S}}{\stack{\sigma_R}}
    \land \tS{R}{\sigma_R}{\sigma_R'} \\
&\qquad\implies
    \exists \sigma_S'.\;
        \rS{S}{\sigma_S}{\sigma_S'}
        \land \subsumes{\stack{\sigma_S'}}{\stack{\sigma_R'}}
\end{aligned}
\]
\end{definition}

\begin{lemma}[Simulation preserves observations]
\label{lemma:sim-inherits-obs}
For solvers $R$ and $S$, if $\simulates{R}{S}$, $\observable{\sigma^{init}_S} \subseteq \observable{\sigma^{init}_R}$.
\end{lemma}
\begin{proof}
Given $\tuple{E_o, I_o} \in \observable{{\sigma^{init}_S}}$, we know there is a $\sigma_S$ with $\rS{S}{\sigma^{init}_S}{\sigma_S}$ and $E_o \subseteq E(\sigma_S) \land I_o \subseteq I(\sigma_S)$.
We have $\subsumes{\stack{\sigma^{init}_R}}{\stack{\sigma^{init}_S}}$ since both trails are empty.
Since $\simulates{R}{S}$, we can simulate every step in $\rS{S}{\sigma^{init}_S}{\sigma_S}$ to get a $\sigma_R$ with $\rS{R}{\sigma^{init}_R}{\sigma_R}$ and $\subsumes{\stack{\sigma_R}}{\stack{\sigma_S}}$.
Then we have $E_o \subseteq E(\sigma_R) \land I_o \subseteq I(\sigma_R)$ and so $\tuple{E_o, I_o} \in \observable{{\sigma^{init}_R}}$.
\end{proof}

\begin{lemma}[Bisimulation with \idealSolver implies confluence]
\label{lemma:bisim-implies-confluence}
For solver $S$, if $\simulates{S}{\idealSolver}$ and $\simulates{\idealSolver}{S}$, $S$ is confluent.
\end{lemma}
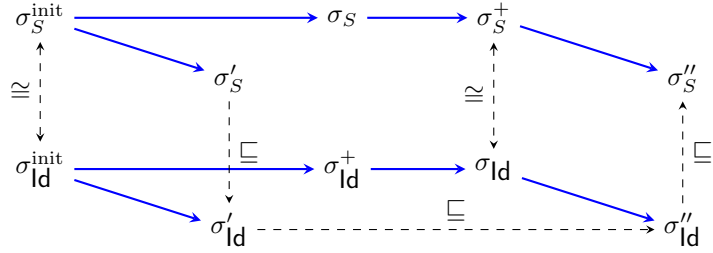
\begin{center}
\begin{tikzpicture}[ 
    >=stealth, 
    state/.style={}, 
    trans/.style={blue,->,thick}, 
    eq/.style={dashed,->}, 
] 
\node[state] (s0) at (0,2) {$\sigma^{\mathrm{init}}_{S}$};
\node[state] (s1) at (4,2) {$\sigma_{S}$};
\node[state] (s1p) at (6,2) {$\sigma^{+}_{S}$};

\node[state] (sp) at (2.5,1.2) {$\sigma'_{S}$};
\node[state] (l0) at (0,0) {$\sigma^{\mathrm{init}}_{\mathrm{\idealSolverID}}$};
\node[state] (l1) at (4,0) {$\sigma^{+}_{\mathrm{\idealSolverID}}$};
\node[state] (l2) at (6,0) {$\sigma_{\mathrm{\idealSolverID}}$};

\node[state] (lp) at (2.5,-0.8) {$\sigma'_{\mathrm{\idealSolverID}}$};
\node[state] (spp) at (8.5,1.2) {$\sigma''_{S}$};
\node[state] (lpp) at (8.5,-0.8) {$\sigma''_{\mathrm{\idealSolverID}}$};

\draw[trans] (s0) -- (s1);
\draw[trans] (s1) -- (s1p);
\draw[trans] (l0) -- (l1);
\draw[trans] (l1) -- (l2); 

\draw[trans] (s0) -- (sp);
\draw[trans] (l0) -- (lp);
\draw[trans] (s1p) -- (spp);
\draw[trans] (l2) -- (lpp); 

\draw[eq,<->] (s0) -- node[left] {$\cong$} (l0);
\draw[eq,->] (sp) -- node[right] {$\sqsubseteq$} (lp);
\draw[eq,<->] (s1p) -- node[left] {$\cong$} (l2);
\draw[eq,->] (lpp) -- node[right] {$\sqsubseteq$} (spp); 

\draw[eq,->] (lp) -- node[above] {$\sqsubseteq$} (lpp); 
\end{tikzpicture}
\end{center}

\begin{proof}
Given $\rS{S}{\sigma^{init}_S}{\sigma_S}$, and a $\tuple{E_o, I_o} \in \observable{\sigma^{init}_S}$, it suffices to show $\tuple{E_o, I_o} \in \observable{\sigma_S}$, since we trivially have $\observable{\sigma_S} \subseteq \observable{\sigma^{init}}$.

$\subsumes{\stack{\sigma^{init}_{\idealSolverID}}}{\stack{\sigma^{init}_S}}$ since both stacks are empty.
Then through $\simulates{\idealSolver}{S}$, we have a $\sigma_{\idealSolverID}^+$ with $\rS{\idealSolverID}{\sigma^{init}_{\idealSolverID}}{\sigma_{\idealSolverID}^+}$ and $\subsumes{\stack{\sigma_{\idealSolverID}^+}}{\stack{\sigma_S}}$.
Since \idealSolver can fully backtrack and push the same equalities with fewer instantiations, we also have a $\sigma_{\idealSolverID}$ with $\rS{\idealSolverID}{\rS{\idealSolverID}{\sigma^{init}_{\idealSolverID}}{\sigma_{\idealSolverID}^+}}{\sigma_{\idealSolverID}}$, $I(\sigma_S) = I(\sigma_{\idealSolverID})$, and $E(\sigma_S) \subseteq E(\sigma_{\idealSolverID})$. And since pushing equalities is always allowed, $\sigma_S$ can transition to a $\sigma^{+}_{s}$ that includes the missing equalities, and so $\subsumes{\stack{\sigma^{+}_{s}}}{\stack{\sigma_{\idealSolverID}}}$

By the definition of $\observable{}$, there is a $\sigma_S'$ with $\rS{S}{\sigma^{init}_S}{\sigma_S'}$, $E_o \subseteq E(\stack{\sigma_S'})$ and $I_o \subseteq I(\stack{\sigma_S'})$.
Since $\simulates{\idealSolver}{S}$, we have a $\sigma_{\idealSolverID}'$ with $\rS{\idealSolverID}{\sigma^{init}_{\idealSolverID}}{\sigma_{\idealSolverID}'}$ and $\subsumes{\stack{\sigma_{\idealSolverID}'}}{\stack{\sigma_S'}}$.

Since \idealSolver is confluent, we have a $\sigma_{\idealSolverID}''$ with $\rS{\idealSolverID}{\sigma_{\idealSolverID}}{\sigma_{\idealSolverID}''}$ and $E(\stack{\sigma_{\idealSolverID}''}) \supseteq E(\stack{\sigma_{\idealSolverID}'}) \supseteq E_o$ and similarly $I(\stack{\sigma_{\idealSolverID}''}) \supseteq I_o$.
Then, since $\simulates{S}{\idealSolver}$ and $\subsumes{\stack{\sigma_S}}{\stack{\sigma_{\idealSolverID}}}$, there is a $\sigma_S''$ with $\rS{S}{\sigma_S}{\sigma_S''}$ and $\subsumes{\stack{\sigma_s''}}{\stack{\sigma_{\idealSolverID}''}}$.
Then $E_o \subseteq E(\stack{\sigma_{\idealSolverID}''})$ and $I_o \subseteq I(\stack{\sigma_{\idealSolverID}''})$ and so $\tuple{E_o,I_o} \in \observable{\sigma_S}$.

\end{proof}

\subsection{Refinement 1: A Naive Solver -- \naiveSolver}
The most straightforward (and slow) approach is to compute all generation numbers at matching time using $I(\tau)$, the current instance set.

The naive solver maintains the same state as the ideal. It uses $\enabled{E}{I}$ instead of $\enabled{E}{\instUniv}$:
\[
\begin{array}{c}
\AxiomC{$\qInst{q}{\instantiatedTerm} \in \enabled{E}{I}$}
\RightLabel{($\PushInst_{\naiveSolverID}$)}
\UnaryInfC{$\tNaive{\sNaive}{\state{\tau \cdot \instEvent{\qInst{q}{\instantiatedTerm}}}}$}
\DisplayProof
\end{array}
\]
All other transitions remain unchanged.

\newcommand{\IH}{\mathrm{IH}}

\begin{theorem}
\label{theorem:n-sim-id}
$\simulates{\naiveSolver}{\idealSolver}$.
\end{theorem}
\begin{proof}
\PushEq can be trivially simulated by the corresponding step in \naiveSolver.
\PopEq and \PopInst can be simulated without any transitions.
\PushInst can be simulated by a sequence of $\PushInst_{\naiveSolverID}$ steps: one for each quantifier instance used in the inference of the matching term's ideal generation, starting at the leaves. We present a detailed proof for the \PushInst case in Appendix \ref{appendix:proof:n-sim-id}.
\end{proof}

\begin{theorem}
$\simulates{\idealSolver}{\naiveSolver}$
\end{theorem}
\begin{proof} 
$\PushEq_{\naiveSolverID}$ and $\PushInst_{\naiveSolverID}$ can each be simulated by the corresponding step in \idealSolver. $\PopEq_{\naiveSolverID}$ and $\PopInst_{\naiveSolverID}$ can be simulated without any transitions.
\end{proof}

\subsection{Refinement 2: Generation-Tracking Solver -- \trackingSolverID}
The next step is to avoid recomputing $\gen{E}{I}{t}$ from scratch at matching time.
Instead, we maintain a dynamic mapping from terms to generation numbers.

Generation tracking states are of the form
$\trackingSolverSig : \state{\tau, g}$
where 
\begin{itemize}
    \item $\tau, E, I$ have the same meaning as in $\idealSolverSig$
    \item $\genMap: \Term \to \Nat$ - mapping used to track generation numbers
\end{itemize}
Initially, $\genMap(t)=0$ if $t \in A$, and $\genMap(t)=\infty$ otherwise.

\subsubsection{Handling equivalent instances:}
Real-world solvers avoid creating equivalent instantiations: given $\instantiatedTerm_1, \instantiatedTerm_2$ with $\cg{\instantiatedTerm_1}{\instantiatedTerm_2}$, for any $q$, the solver instantiates at most one of $\qInst{q}{\instantiatedTerm_1}$ or $\qInst{q}{\instantiatedTerm_2}$, since it judges the two instances to be equivalent. This complicates generation tracking, as the generation updates resulting from a single quantifier instance must emulate those of the whole set of equivalent instances it stands in for.

Our model allows equivalent instantiations, so this requirement could be ignored without sacrificing completeness or conservativeness. We nonetheless reflect it in the design of the generation-tracking solver, through the definition of tracked generation:

\begin{definition}[Tracked generation]
The generation for term $t$ used by $\trackingSolver$ is
\[
\genTracked{E}{\genMap}{t} \defineEq \max\{ \min \{\genMap(c) \mid \cg{c}{s}\} \mid \isSubterm{s}{t}\}
\]
\end{definition}
\begin{definition}[Enabled instances for \trackingSolver]
\[
\enabledTr{E}{\genMap} \defineEq \{\qInst{q}{\instantiatedTerm} \mid \weight{q} + \genTracked{E}{\genMap}{\instantiatedTerm} \leq \maxCost, 
q \in Q, \instantiatedTerm \mbox{ is a term}\}
\]
\end{definition}

\begin{definition}[Generation updates] 
A quantifier instance $\qInst{q}{\instantiatedTerm}$ causes $\genMap$ to be updated as follows:
\[
\updateTr{E}{\genMap}{\qInst{q}{\instantiatedTerm}} \defineEq 
\mapMultiUpdate{g}{t}
{\min(\genMap(t), \weight{q} + \genTracked{E}{\genMap}{\instantiatedTerm})}
{t \in \qInst{q}{\instantiatedTerm}}
\]
\end{definition}

Generation updates are sticky, meaning we do not undo them when we pop instances.
Doing so would not only require additional backtracking state, but also make the solver re-discover the same updates later, which is expensive in practice.

\subsubsection{Modified transitions and definitions}
\[
\begin{array}{c}
\AxiomC{$\qInst{q}{\instantiatedTerm} \in \enabledTr{E}{\genMap}$}
\RightLabel{(\PushInst$_{\trackingSolverID}$)}
\UnaryInfC{$\tTr{\sTr}{\state{\tau \cdot \instEvent{\qInst{q}{\instantiatedTerm}}, \updateTr{E}{\genMap}{\qInst{q}{\instantiatedTerm}}}}$}
\DisplayProof
\end{array}
\]

All other transitions remain unchanged except state augmentation with $g$.

\begin{theorem}
\label{theorem:t-sim-naive}
$\simulates{\trackingSolver}{\naiveSolver}$.
\end{theorem}
\begin{proof} Simulating \PushEq, \PopEq, and \PopInst is trivial (see the proof of Theorem \ref{theorem:n-sim-id}).
Lemma \ref{lemma:tr-finds-generations} in the appendix shows that $\PushInst_{\naiveSolverID}$ can be simulated by a sequence of $\PushInst_{\trackingSolverID}$ steps.
\end{proof}

\begin{theorem}
$\simulates{\idealSolver}{\trackingSolver}$.
\end{theorem}
\begin{proof}
Since
$\PushEq_{\trackingSolverID}$, $\PopEq_{\trackingSolverID}$, and $\PopInst_{\trackingSolverID}$ are unchanged, simulating them remains trivial. By \lemmaRef{trackedgen-lowerbound} in the appendix, there exists an $E'$ such that $\forall t. \; \genideal{E'}{t} \leq  \genTracked{E(\tau)}{g}{t}$. A $\PushInst_{\trackingSolverID}$ step can be simulated by first pushing all the equalities in $E'$,  then taking a $\PushInst$ step.
\end{proof}

\subsection{Refinement 3: Congruence-Class-Tracking Solver -- $\cgcSolverID$}
\trackingSolver is still inefficient because computing $\min \{\genMap(c) \mid \cg{c}{t}\}$ is expensive for large congruence classes (\secref{eval} shows this empirically). The next version avoids iterating over congruence classes at match time by maintaining one generation number per congruence class.

Solver state is now a tuple
$\cgcSolverSig : \sCGC$
where
\begin{itemize}
    \item $E, I$ have the same meaning as in $\naiveSolverSig$ and $\trackingSolverSig$    
    \item $\cgcMap: \Term \to \Set(\Term)$ - maps a term to its congruence class
    \item $\genMapCGC: \Set(\Term) \to \Nat$ - maps a congruence class to its generation number.
    \item We replace $\eqEvent{e}$ by $\eqMergeEvent{e}{\cgcMap}{\genMapCGC}$ to model backtracking to a snapshot
    \item $\stickyUpdates: \Set(\Term \times \Nat)$ - cumulative set of sticky generation updates
\end{itemize}
Initially,  
\begin{itemize}
    \item $\stickyUpdates = \emptyset$
    \item $\cgcMap(t) = \{t\}$, for all $t$.
    \item $\genMapCGC(\{t\}) = 0$ if $t \in A$ and $\genMapCGC(\cdot) = \infty$ otherwise, for all $t$.
\end{itemize}
 
\cgcSolver needs to explicitly keep track of congruence classes, which SMT solvers readily do. 
Congruence classes are merged, modifying $\cgcMap$, as a result of \PushEq{} and unmerged as a result of \PopEq. We make sure $\genMapCGC$ is updated accordingly, such that it always maps to the generation of the minimal term in the class.

\begin{definition}[Merging congruence classes]
To track merging of congruence classes and updates to generations we
introduce the relation 
$\tCGCMerge
    {\tuple{\cgcMap, \genMapCGC}}
    {E}
    {\tuple{\cgcMap', \genMapCGC'}}
    $.
A single merge step $\rightarrow_{\mathit{merge}(E)}$ is given by the inference:
\[
\begin{array}{c}

\AxiomC{$\cg{t_1}{t_2}$}
\AxiomC{$X = \cgcMap(t_1), Y = \cgcMap(t_2)$}
\AxiomC{$\cgcMap'=\mapMultiUpdate{\cgcMap}{t}{X \cup Y}{t \in X \cup Y}$}
\TrinaryInfC{$
\tuple{\cgcMap, \genMapCGC}
\rightarrow_{\mathit{merge}(E)}
\tuple{
    \cgcMap', 
    \genMapCGC[X \cup Y \mapsto \min(\genMapCGC(X), \genMapCGC(Y))]}$}
\DisplayProof
\end{array}
\]
The relation
$\tCGCMerge
    {\tuple{\cgcMap, \genMapCGC}}
    {E}
    {\tuple{\cgcMap', \genMapCGC'}}$
applies $\rightsquigarrow_E$ exhaustively until
all congruent terms belong to the same class.
\end{definition}

\begin{definition}[Congruence-class-tracked generation]
The effective generation of a term $t$ for $\cgcSolver$ is:
\[
\genCGC{\cgcMap}{\genMapCGC}{t} \defineEq \max\{  \genMapCGC(\cgcMap(s)) \mid s\in t\}
\]
\end{definition}

$\genMapCGC$ is analogous to $\genMap$ in $\trackingSolver$. It is therefore updated similarly, in response to quantifier instantiations:
\begin{definition}[Generation updates for \cgcSolver]
\[
\updateCGC{\cgcMap}{\genMapCGC}{\qInst{q}{\instantiatedTerm}} 
\defineEq 
\mapMultiUpdate{\genMapCGC}{\cgcMap(t)}
{\min(\genMapCGC(\cgcMap(t)), \weight{q} + \genCGC{\cgcMap}{\genMapCGC}{\instantiatedTerm})}
{t\in \qInst{q}{\instantiatedTerm}}
\]
\end{definition}

One complication that arises is that naively backtracking congruence-class merges may undo generation updates resulting from quantifier instantiation, which are meant to be sticky. It is crucial to prevent this, as re-discovering lower generations for existing terms is expensive in practice. To this end, we explicitly store a cumulative set of generation updates, which are re-applied on \PopEq.

\begin{definition}[New sticky updates caused by an instantiation]
Each term $t$ re-discovered at a lower generation through a quantifier instance $\qInst{q}{\instantiatedTerm}$ is tagged with a sticky update in $\newStickyUpdates{\qInst{q}{\instantiatedTerm}}{\cgcMap}{\genMapCGC}{\genMapCGC'}$:
\[
\newStickyUpdates{\qInst{q}{\instantiatedTerm}}{\cgcMap}{\genMapCGC}{\genMapCGC'} \defineEq \{\tuple{t, \genMapCGC'(\cgcMap(t))} \mid t\in \qInst{q}{\instantiatedTerm} \land \genMapCGC'(\cgcMap(t)) < \genMapCGC(\cgcMap(t)) < \infty\}
\]
\end{definition}

\begin{definition}[Replaying sticky updates]
Applying the updates stored in $\stickyUpdates$ to $\sCGC$ leads to $\state{\tau, \cgcMap, \genMapCGC', \stickyUpdates}$ if 
$\tApplyUpdates{\cgcMap}{\stickyUpdates}
    {\genMapCGC}
    {\genMapCGC'}$ can be obtained from one of the inferences:
\[
\begin{array}{c}
\AxiomC{}
\RightLabel{(\text{\mbox{ApplyNone}})}
\UnaryInfC{$
\tApplyUpdates{\cgcMap}{\emptyset}
    {\genMapCGC}
    {\genMapCGC}$}
\DisplayProof
\\
\\
\AxiomC{$
\tApplyUpdates{\cgcMap}{\stickyUpdates}
    {\genMapCGC}
    {\genMapCGC'}$}
\RightLabel{(ApplySome)}
\UnaryInfC{$
\tApplyUpdates{\cgcMap}{\stickyUpdates
\cup \{\tuple{t, n}\}}
    {\genMapCGC}
    \genMapCGC'[(\cgcMap(t))
    \mapsto 
    \min(\genMapCGC'(\cgcMap(t)), n)]$}
\DisplayProof
\end{array}
\]
\end{definition}

\begin{definition}[Enabled instances for \cgcSolver]
\[
\enabledCGC{\cgcMap}{\genMapCGC}  
\defineEq
\{\qInst{q}{\instantiatedTerm} \mid \weight{q} + \genCGC{\cgcMap}{\genMapCGC}{\instantiatedTerm} \leq \maxCost, q \in Q, \instantiatedTerm \mbox{ is a term}\}
\]
\end{definition}

We now present the transition rules for \cgcSolver:

\[
\begin{array}{c}
\AxiomC{$e \defineEq (t_1 \simeq t_2)$}
\AxiomC{$\tCGCMerge{\tuple{\cgcMap, \genMapCGC}}{E \cup \{e\}}{\tuple{\cgcMap', \genMapCGC'}}$}
\RightLabel{(\PushEq$_{\cgcSolverID}$)}
\BinaryInfC{$\tCGC{\sCGC}{\state{\tau\cdot\eqMergeEvent{e}{\cgcMap}{\genMapCGC}, \cgcMap', \genMapCGC', \stickyUpdates}}$}
\DisplayProof
\\
\\
\AxiomC{$\tApplyUpdates{\cgcMap'}{\stickyUpdates}{\genMapCGC'}{\genMapCGC''}$}
\RightLabel{(\PopEq$_{\cgcSolverID}$)}
\UnaryInfC{$
\tCGC
    {\state{\tau \cdot \eqEvent{e, \cgcMap', \genMapCGC'}, \cgcMap, \genMapCGC, \stickyUpdates}}
    {\state{\tau, \cgcMap', \genMapCGC'', \stickyUpdates}}$}
\DisplayProof
\\
\\
\AxiomC{$\qInst{q}{\instantiatedTerm} \in \enabledCGC{\cgcMap}{\genMapCGC}$}
\AxiomC{$\genMapCGC'=\updateCGC{\cgcMap}{\genMapCGC}{\qInst{q}{\instantiatedTerm}}$}
\RightLabel{(\PushInst$_{\cgcSolverID}$)}
\BinaryInfC{$
\tCGC
    {\sCGC}
    {\state{ \tau \cdot \instEvent{\qInst{q}{\instantiatedTerm}}, \cgcMap, \genMapCGC', \stickyUpdates \cup \newStickyUpdates{\qInst{q}{\instantiatedTerm}}{\cgcMap}{\genMapCGC}{\genMapCGC'} }}$}
\DisplayProof
\\
\\
\AxiomC{$$}
\RightLabel{(\PopInst$_{\cgcSolverID}$)}
\UnaryInfC{$\tCGC{\state{ \tau\cdot \instEvent{\qInst{q}{\instantiatedTerm}}, \cgcMap, \genMapCGC, \stickyUpdates}} {\state{ \tau, \cgcMap, \genMapCGC, \stickyUpdates}}$}
\DisplayProof
\\
\\
\end{array}
\]

\begin{lemma}
\label{lemma:cgc-tracking-eq}
Every transition sequence executable from the initial state of \(\cgcSolver\) is executable from the initial state of \(\trackingSolver\), and vice versa; moreover, the resulting states satisfy
\[
\forall t. \; \genMapCGC(\cgcMap(t)) = \min\{\genMap(c) \mid \cg{c}{t}\}
\]
\end{lemma}
\begin{proof}
By induction on the length of the transition sequence. The invariant holds initially, since both maps assign generation $0$ to input terms and $\infty$ to all other terms.

On \PushEq, merging congruence classes updates $\genMapCGC$ to the minimum of the two merged classes' generations. This is the same as $\min\{\genMap(c) \mid \cg{c}{t}\}$ once the classes are merged.

On \PushInst, $\cgcSolver$ updates $\genMapCGC(\cgcMap(t))$ for every subterm $t$ in the instance using the same weight and, by the induction hypothesis, the same generation for $\instantiatedTerm$ as $\trackingSolver$ uses to update $\genMap$.

On \PopEq$_{\cgcSolverID}$, symmetrically to \PushEq$_{\cgcSolverID}$, unmerging congruence classes restores the minimum over the resulting classes. The sticky-update mechanism ensures no generation updates are lost in the process: any update to a term's generation, once recorded as a sticky update, survives the split and is re-applied.

\PopInst$_{\cgcSolverID}$ does not affect either map. So the invariant is trivially preserved.
\end{proof}

\begin{theorem}
$\simulates{\cgcSolver}{\trackingSolver}$.
\label{theorem::cgc-sim-tracking}
\end{theorem}
\begin{proof}
Given $\sigma_{\trackingSolver}$ and $\sigma_{\cgcSolver}$, the latter can simulate a step to $\sigma_{\trackingSolver}'$ by first replaying the transition sequence $\rS{\trackingSolver}{\sigma^{init}_{\trackingSolver}}{\sigma_{\trackingSolver}}$ and then taking the same final step. The sequence can be replayed due to \lemmaRef{cgc-tracking-eq}, and the fact that the value of $\genMapCGC(\cgcMap(t))$ at  $\sigma_{\cgcSolver}$ is no greater than its value at $\sigma^{init}_{\cgcSolver}$ for all $t$.
\end{proof}

\begin{theorem}
$\simulates{\trackingSolver}{\cgcSolver}$.
\end{theorem}
\begin{proof}
The proof is symmetric to that of Theorem \ref{theorem::cgc-sim-tracking}.
\end{proof}

\section{Z3 Implementation}
\label{s:implementation}
We implemented \cgcSolver in Z3; our implementation has been upstream since version 5.0.
This section briefly describes some of the implementation details.

Z3 associates each internalized term with a metadata structure called an e-node.
It keeps track of congruence classes through a hash table which stores a unique e-node per congruence class, called the \textit{congruence representative}.
Z3 uses this table during pattern matching to ensure it only binds congruence representatives, in order to avoid duplicate matches.
This table effectively stores $\cgcMap$ from \cgcSolver.

One option, in line with \trackingSolverID, would be to store the function $g$ as an e-node field set at internalization and updated during quantifier instantiation. At matching time, we could walk the congruence class of each subterm of the match candidate to determine the minimum generation. We experimented with this option but discarded it, as it led to significant performance regression (see \secref{eval}).

Another option, in line with \cgcSolver, would be to forego $\genMap$ and instead store $\genMapCGC$, directly in the hash table entry. This performs much better than the prior option, but still incurs a minor (5\%) slowdown, since looking up the hash table is less cache-friendly than reading the generation number directly off the e-node.

Our implementation instead stores $\genMapCGC$ on the e-node of the congruence root.
When congruence classes get merged, we update the new root's e-node and save the old generation on the backtracking stack.
This wastes some space, since the generation field goes unused on every e-node that is not a congruence root, but the tradeoff is worthwhile: in the common case the congruence root's e-node is already accessed before a generation lookup, so we avoid the extra access.

We store $\delta$ explicitly as a list of sticky updates to be re-applied whenever congruence classes get unmerged.
This can be expensive: sticky updates survive backtracking until the term whose generation was updated is garbage collected, and re-applying all of them costs time linear in the number of updates.
In practice, however, discovering an existing term at a lower generation is uncommon enough (despite its significance for instability) that this cost is negligible.

Propagating generation updates is another potentially expensive consequence of our implementation. When a term whose generation previously dominated other terms' generations (via instantiation) is itself updated, that update must propagate, e.g., by allowing duplicate matches with lower generation numbers. Our implementation skips this, slightly deviating from \cgcSolver, yet still greatly improves stability. We suspect enabling such duplicate matches would not significantly affect performance for the same reason as above: generation updates are rare in practice.

\section{Empirical Evaluation}
\label{s:eval}
\textbf{Experimental setup:}
For our evaluation, we use the Mariposa~\cite{mariposa} framework, and its unstable-core benchmark set.
This consists of 60 queries from prior program verification projects, mostly written in Dafny~\cite{dafny}.
Mariposa stability-fuzzes these queries by generating mutants through reseeding, assertion shuffling and variable renaming.

For each query, we generate 99 mutants and run each of them, along with the original query, with a 60-second timeout (Mariposa's default). We include queries that are not outcome-unstable to ensure we do not introduce regressions that increase timeouts.
Our experiments were run on a machine with two 64-core AMD EPYC 7763 CPUs at 2.4GHz and 503GiB of RAM. For two of the queries, running all 100 variants in parallel allocates enough memory to crash our infrastructure; three others produce mutants that return a solver error due to malformed input. We discard these five queries and present results for the remaining 55 queries, producing 5500 mutants, each of which returns \unsat, \unknown, or \timeout.

\newcommand{\implT}{\cc{z3-\trackingSolver}}
\newcommand{\implCHashtable}{\cc{z3-\cgcSolver-table}}
\newcommand{\implCEnodes}{\cc{z3-\cgcSolver-enode}}

We compare four versions of generation accounting: z3's accounting prior to our work, which we treat as a baseline, the two less efficient options outlined in \secref{implementation} (\implT, \implCHashtable), and our final implementation (\implCEnodes). 

\begin{table}[!h]
\centering
\setlength{\tabcolsep}{6pt}
\begin{tabular}{lccc}
\toprule
\textbf{Solver version} & \makecell{\textbf{Outcome-unstable}\\\textbf{queries}} & \makecell{\textbf{Unknown}\\\textbf{mutants}} & \makecell{\textbf{Timeouts}} \\
\midrule
\cc{baseline}    & 27 & 489 & 1813 \\
\implT           & 6  & 25  & 2399 \\
\implCHashtable  & 4  & 25  & 1867 \\
\implCEnodes     & 4  & 29  & 1779 \\
\bottomrule
\end{tabular}

\vspace{8pt}

\caption{Stability and performance results on Mariposa's unsat core benchmark across implementation versions. Experiments include 55 queries generating 5500 mutants.}
\label{tab:mariposa-eval}
\end{table}

Table~\ref{tab:mariposa-eval} reports, for each version, the number of outcome-unstable queries (i.e., queries that have at least one \unknown and one \unsat mutant), the total number of \unknown mutants across all queries, and the total number of timeouts.

\textbf{Stability}:
All three implementations substantially reduce outcome instability relative to the baseline: \implCEnodes cuts the number of unknown mutants by 94\% and the number of outcome-unstable queries by 85\%, with \implT and \implCHashtable{} achieving comparable reductions; the differences are within noise.

\textbf{Performance}:
We also observe a performance cost, measured as the total number of timeouts across all queries: \implT is considerably slower, producing 32\% more timeouts than the baseline, while \implCHashtable incurs only a minor regression (3\% more timeouts). \implCEnodes, our final implementation, avoids this cost and yields slightly fewer timeouts than the baseline.
Additionally, for each solver version, we manually examined the timeout count per query and confirmed that our changes did not cause any single query to blow up.

To ensure our changes do not introduce regressions beyond program verification, we also ran \implCHashtable{} and \implCEnodes{} on the non-incremental SMT-LIB benchmarks~\cite{smtlib2025}, using a 20-second timeout. For quantified categories, all differences fall within noise, the largest being 14 additional timeouts (1\%) for \implCHashtable{} on ALIA. Among quantifier-free categories, the only notable difference is a 5\% slowdown for \implCHashtable{} on QF\_UF, which does not cause additional timeouts. This is expected: the same cache-locality issue affecting Mariposa performance also slows down congruence closure in the absence of quantifiers. \implCEnodes{} avoids this regression entirely.

\bibliographystyle{splncs04}
\bibliography{biblio}

\appendix

\section{Proofs related to \naiveSolver}
\label{appendix:proof:n-sim-id}

\newcommand{\tSInst}[3]{\transitionSup{#2}{#1}{inst}{#3}}
\newcommand{\tIDInst}[2]{\tSInst{\idealSolverID}{#1}{#2}}
\newcommand{\tNInst}[2]{\tSInst{\naiveSolverID}{#1}{#2}}
\newcommand{\tTInst}[2]{\tSInst{\trackingSolverID}{#1}{#2}}

\newcommand{\rNInst}[2]{\reachableSup{#1}{\naiveSolverID}{inst}{#2}}
\newcommand{\rTInst}[2]{\reachableSup{#1}{\trackingSolverID}{inst}{#2}}

\begin{lemma}[\naiveSolver can find all ideal generations within threshold]
\label{lemma:naive-computes-ideal}
\[
    \genideal{E}{t} \leq \maxCost
    \implies \\
    \exists \tau'. 
        \rNInst{\state{\tau}}{\state{\tau'}} 
        \land \subsumes{\tau'}{\tau}
        \land \gen{E}{I(\tau')}{t} \leq \genideal{E}{t}
\]
where $\rNInst{}{}$ denotes reachability through exclusively $\PushInst_{\naiveSolverID}$ transitions.
\end{lemma}
\begin{proof}
Since $\genideal{E}{\instantiatedTerm} \leq G < \infty$,  we have an inference for $\derivable{E}{\instUniv}{t}{\genideal{E}{t}}$. By induction on this inference, we have three cases:
\begin{itemize}
    \item (\Input) $t \in A$: 
     Then $\gen{E}{I}{t} = 0 = \genideal{E}{t}$. So $\tau' \defineEq \tau$ is a witness.
    \item (\Congruence) $t=t'[s_2], s_1 \cong_{E} s_2, \derivable{E}{\instUniv}{t'[s_1]}{n_{t'}}, \derivable{E}{\instUniv}{s_2}{n_{s_2}}, \genideal{E}{t}=\max(n_{t'}, n_{s_2})$:
    We have $n_{t'} \leq \maxCost$ and $n_{s_2} \leq \maxCost$ and so, by the induction hypotheses, we get
    $\rNInst{\state{\tau}}{\state{\tau_1}} \land \gen{E}{I(\tau_1)}{t[s_1]} \leq n_{t'}$ and $\rNInst{\state{\tau}}{\state{\tau_2}} \land \gen{E}{I(\tau_2)}{s_2} \leq n_{s_2}$.
    Since adding instances never blocks $\PushInst_{\naiveSolverID}$, we can replay the transitions in $\rNInst{\state{\tau}}{\state{\tau_2}}$ from  $\state{\tau_1}$ and get a $\tau'$ with $\rNInst{\state{\tau}}{\rNInst{\state{\tau_1}}{\state{\tau'}}}$, $\subsumes{\tau'}{\tau_1}$, and $\subsumes{\tau'}{\tau_2}$.
    Then $\gen{E}{I(\tau')}{t[s_1]} \leq \gen{E}{I(\tau_1)}{t[s_1]} \leq n_{t'}$ and similarly $\gen{E}{I(\tau')}{s_2} \leq n_{s_2}$.
    So, through the \Congruence rule, $\gen{E}{I(\tau')}{t} \leq n = \genideal{E}{t}$.
    \item (\Inst) $t \in \qInst{q}{\instantiatedTerm}, \genideal{E}{t} = \weight{q} + n_{\instantiatedTerm}, \derivable{E}{\instUniv}{\instantiatedTerm}{n_{\instantiatedTerm}}$:
    Since quantifier weights are non-negative, $n_{\instantiatedTerm} \leq n \leq \maxCost$.
    Then by the induction hypothesis we have a $\tau''$ with $\gen{E}{I(\tau'')} {\instantiatedTerm} \leq \genideal{E}{\instantiatedTerm} \land \rNInst{\state{\tau}}{\state{\tau'}} 
    \land \subsumes{\tau'}{\tau}$. 
    $\PushInst_{\naiveSolverID}$ with $\qInst{q}{\instantiatedTerm}$ is enabled in $\state{\tau''}$ and leads to some $\state{\tau'}$ such that $\tau'$ is a witness.
\end{itemize}
\end{proof}
    
\begin{lemma} [The \PushInst case of Theorem \ref{theorem:n-sim-id}]
For all $\sigma_{\naiveSolverID}, \sigma_{\idealSolverID}, \sigma_{\idealSolverID}'$,
\[
    \subsumes{\stack{\sigma_{\naiveSolverID}}}{\stack{\sigma_{\idealSolverID}}} \land \tIDInst{\sigma_{\idealSolverID}}{\sigma_{\idealSolverID}'} \implies \\
    \exists \sigma_{\naiveSolverID}'. \; \rS{\naiveSolverID}{\sigma_{\naiveSolverID}}{\sigma_{\naiveSolverID}'} \land \subsumes{\stack{\sigma_{\naiveSolverID}'}}{\stack{\sigma_{\idealSolverID}'}}
\]
\end{lemma}
\begin{proof}
Let $E$ be $E(\stack{\sigma_{\idealSolverID}})$ (or equivalently $E(\stack{\sigma_{\naiveSolverID}})$, since $\subsumes{\stack{\sigma_{\naiveSolverID}}}{\stack{\sigma_{\idealSolverID}}}$).
From $\tIDInst{\sigma_{\idealSolverID}}{\sigma_{\idealSolverID}'}$ we have $\genideal{E}{\instantiatedTerm} + \weight{q} \leq \maxCost$, where $\qInst{q}{\instantiatedTerm}$ is the instance being pushed.
By Lemma \ref{lemma:naive-computes-ideal}, we have a $\tau'$ with $\gen{E}{I(\tau')} {\instantiatedTerm} \leq \genideal{E}{\instantiatedTerm} \land \rNInst{\sigma_{\naiveSolverID}}{\state{\tau'}} \land \subsumes{\tau'}{\stack{\sigma_{\naiveSolverID}}}$
$\PushInst_{\naiveSolverID}$ with $\qInst{q}{\instantiatedTerm}$ is enabled in $\state{\tau'}$ and leads to some $\sigma_{\naiveSolverID}'$ with $\subsumes{\stack{\sigma_{\naiveSolverID}'}}{\stack{\sigma_{\idealSolverID}'}}$.
\end{proof}

\section{Proofs related to $\trackingSolver$}

\begin{lemma}[\trackingSolver can find derivable generations]
\label{lemma:tr-finds-generations}
For all $t, n, \tau, \genMap$ such that $\rS{\trackingSolverID}{\sigma^{init}_{\trackingSolverID}}{\state{\tau, \genMap}}$,  
\[
\derivable{E}{I}{t}{n}
\implies
\exists \tau', \genMap'. \; 
    \rTInst{\state{\tau, \genMap}}{\state{\tau', \genMap'}} 
    \land
    \genTracked{E}{\genMap'}{t} \leq n
\]
\end{lemma}
\begin{proof}
By induction on $\derivable{E}{I}{t}{n}$, we have three cases:
\begin{itemize}
    \item (\Input) $t \in A, n=0$:
    Since $A$ is subterm-closed, $t' \in A$ for all subterms $t'$ of $t$. Then we have $\genMap^{init}(t') = 0$ at initialization and because $\genMap(t')$ cannot increase through any transition, $\genMap(t') = 0$. So $\tau'\defineEq\tau, \genMap' \defineEq g$ is a witness.
    \item (\Inst) $t \in \qInst{q}{\instantiatedTerm}, \gen{E}{I}{t} = \weight{q} + n_{\instantiatedTerm}, \derivable{E}{I}{\instantiatedTerm}{n_{\instantiatedTerm}}$:
    By the induction hypothesis we have $\rTInst{\state{\tau, \genMap}}{\state{\tau'', \genMap''}}$ and $\genTracked{E}{\genMap''}{t} \leq n_{\instantiatedTerm}$.
    $\PushInst_{\trackingSolverID}$ is enabled at $\state{\tau'', \genMap''}$ and leads to some $\state{\tau', \genMap'}$ that serves as a witness. 
    \item (\Congruence) $t=t'[s_2], s_1 \cong_{E} s_2, \derivable{E}{I}{t'[s_1]}{n_{t'}}, \derivable{E}{I}{s_2}{n_{s_2}}, \gen{E}{I}{t}=\max(n_{t'}, n_{s_2})$:
    Similarly to the argument in the proof of Lemma \ref{lemma:naive-computes-ideal}, we can take the transitions described by the two inductive hypotheses to get a $\tau', \genMap'$ with $\rTInst{\state{\tau, \genMap}}{\state{\tau', \genMap'}}$, $\genTracked{E}{g'}{t'[s_1]} \leq n_{t'}$, and $\genTracked{E}{g'}{s_2} \leq n_{s_2}$.
    It suffices to show $\min \{\genMap'(c) \mid \cg{c}{t_s}\} \leq \gen{E}{I}{t}$ for all subterms $t_s$ of $t'[s_2]$. We consider three cases:
    \begin{itemize}
        \item $\isSubterm{t_s}{s_2}$: Since $\genTracked{E}{g''}{s_2} \leq n_{s_2}$,  $\min \{\genMap'(c) \mid \cg{c}{t_s}\} \leq  n_{s_2} \leq \gen{E}{I}{t}$.
        \item $\isSubterm{t_s}{t'[s_1]}$: The same argument applies with $\genTracked{E}{g''}{t'[s_1]} \leq n_{t'}$
        \item $t_s = t''[s_2], \isSubterm{t''[s_1]}{t'[s_1]}$: We have $\min \{\genMap'(c) \mid \cg{c}{t''[s_1]}\} \leq n_{t'}$ due to $\genTracked{E}{g''}{t'[s_1]} \leq n_{t'}$. Since $\cg{t''[s_1]}{t''[s_2]}$, the same minimal $c$ ensures $\min \{\genMap'(c) \mid \cg{c}{t''[s_2]}\} \leq n_{t'} \leq \gen{E}{I}{t}$.
    \end{itemize}
    
\end{itemize}
\end{proof}

\begin{lemma}[Mapped values for subterms]
\label{lemma:subterm-genmap}
For all reachable states $\state{\tau, \genMap}$ of \trackingSolver, for all terms $t, s$, $s \in t \implies \genMap(s) \leq \genMap(t)$.
\end{lemma}
\begin{proof}
    Since $A$ is subterm-closed, initially either $\genMap(t) = \genMap(s) = 0$ or $\genMap(t) = \genMap(s) = \infty$.
    Then, $\genMap(t)$ and $\genMap(s)$ may only decrease, via $\PushInst_{\trackingSolverID}$ transitions, which never decrease $\genMap(t)$ further than $\genMap(s)$.
\end{proof}

\begin{lemma}[Tracked generations are lower-bounded by some ideal]
\label{lemma:trackedgen-lowerbound}
For all reachable states $\state{\tau, g}$ of \trackingSolver,
\[
    \exists E'. \forall t. \; \genideal{E'}{t} \leq  \genTracked{E(\tau)}{g}{t}
\]
\end{lemma}
\begin{proof}
By induction on the reachability of $\state{\tau, g}$.
\begin{itemize}
    \item Initially, $\tau$ is empty. We choose $E' := \emptyset$. Then for all $t \in A$, both generations are 0 ($\genTracked{E(\tau)}{g}{t} = 0$ because A is subterm-closed), and for $t \not\in A$ both generations are $\infty$.
    \item $(\PushEq)$ $\tau \defineEq\tau'\cdot\eqEvent{e}$,
    $\IH_1 \defineEq \forall t. \; \genideal{E'}{t} \leq  \genTracked{E(\tau')}{g}{t}$:
        We choose the witness $E'' \defineEq E' \cup E(\tau)$. 
        Proof by structural induction on $t$.
        Let $n \defineEq \genTracked{E(\tau)}{g}{t}$.
        \begin{itemize}
            \item If $t$ is a constant (0-arity function), the only term congruent to $t$ is $t$ itself. Then $\genMap(t) = n$ and $\genTracked{E(\tau')}{\genMap}{t} = n$. By $\IH_1$, $\genideal{E'}{t} \leq n$. Since $E'\subseteq E''$, $\genideal{E''}{t} \leq n$.
            \item Otherwise $t \defineEq c[a]$. By the structural induction hypothesis, $\genideal{E''}{a} \leq n$. Assume all terms congruent to $t$ can be expressed as $c[b]$ for some $b$ with $a \cong_{E(\tau)} b$. In reality, a congruent term can be obtained by rewriting more than one argument of the top-level function as well, but we ignore this for simplicity of presentation.
            For some $b$, $\genMap(c[b]) \leq n$. 
            Then, by \lemmaRef{subterm-genmap}, $\genTracked{E(\tau')}{\genMap}{c[b]} \leq n$ and so by $\IH_1$, $\genideal{E'}{c[b]} \leq n$.
            Since $E'\subseteq E''$, $\genideal{E''}{c[b]} \leq n$ and since $E(\tau) \subseteq E'$, $a \cong_{E''} b$.
            Then by the \Congruence rule for deriving generations, $\genideal{E''}{c[a]} \leq n$.
        \end{itemize}
    \item $(\PopEq)$ $\tau\cdot\eqEvent{e} = \tau'$, 
    $\IH_1 \defineEq \forall t. \; \genideal{E'}{t} \leq  \genTracked{E(\tau')}{g}{t}$:
    Since $E(\tau) \subseteq E(\tau')$, for all $t$, $\genTracked{E(\tau')}{g}{t} \leq \genTracked{E(\tau)}{g}{t}$. Then $\genideal{E'}{t} \leq \genTracked{E(\tau)}{g}{t}$ so $E'$ serves as a witness.
    \item $(\PushInst)$ $\tau = \tau'\cdot\instEvent{\qInst{q}{\instantiatedTerm}}, g = \updateTr{E}{g'}{\qInst{q}{\instantiatedTerm}}, \IH_1 \defineEq \forall t. \; \genideal{E'}{t} \leq  \genTracked{E(\tau')}{g'}{t}$:
    We choose $E'' \defineEq E' \cup E(\tau)$ as the witness.
    The proof is structurally inductive, like the \PushEq case.
    Let $n \defineEq \genTracked{E(\tau)}{g}{t}$.
    \begin{itemize}
        \item If $t$ is a constant, $\genMap(t) = n$ and $\genMap'(t) = \genTracked{E(\tau')}{g'}{t}$.
        Since $E'\subseteq E''$, $\genideal{E''}{t} \leq n$.
        If $\genMap(t) = \genMap'(t)$ we conclude by $\IH_1$. 
        Otherwise, $g(t) = \weight{q}  + \genTracked{E}{\genMap'}{\instantiatedTerm}$ with $t \in \qInst{q}{\instantiatedTerm}$.
        Since $E(\tau)\subseteq E''$, $\weight{q} + \genideal{E''}{\instantiatedTerm} \leq n$.
        Then by the \Inst rule for deriving generations, $\genideal{E''}{t} \leq n$.
        \item If $t$ is a function application, we follow an identical argument to the second case of $(\PushEq)$.
        An ideal generation less than $n$ can be derived for $c[b]$ via the \Inst rule, and then for $c[a]$ via the \Congruence rule.
    \end{itemize}
    \item $(\PopInst)$: $E(\tau) = E(\tau')$ and $g = g'$, so we conclude trivially via the induction hypothesis.    
\end{itemize}
\end{proof}

\end{document}